\documentclass[a4paper,twoside,11pt]{amsart}
\usepackage{amsmath}
\usepackage{amsfonts}
\usepackage{amssymb}
\usepackage{amsthm}
\usepackage{newlfont}
\usepackage{graphicx}
\usepackage{amscd}

\theoremstyle{plain}
\newtheorem{Th}{Theorem}[section]
\newtheorem{Cor}[Th]{Corollary}
\newtheorem{Lem}[Th]{Lemma}
\newtheorem{Prop}[Th]{Proposition}
\theoremstyle{definition}
\newtheorem{Def}{Definition}[section]
\newtheorem{Ex}{Example}[section]
\theoremstyle{remark}
\newtheorem*{Rem}{Remark}
\numberwithin{equation}{section}
\newcommand{\PP}{{\mathbb P}}

\newcommand{\CC}{{\mathbb C}}
\newcommand{\DD}{{\mathbb D}}

\newcommand{\ZZ}{{\mathbb Z}}
\newcommand{\NN}{{\mathbb N}}
\newcommand{\QQ}{{\mathbb Q}}

\newcommand{\RR}{{\mathbb R}}

\newcommand{\XX}{{\mathbb X}}

\newcommand{\bphi}{\boldsymbol{\phi}}

\begin{document}

\title{The~pentagon~relation and incidence geometry}

\author[A. Doliwa]{Adam Doliwa}
\address{A. Doliwa: Institute of Mathematics, Polish Academy of Sciences,
ul.~\'{S}niadeckich 8, 00-956 Warszawa, Poland}
\email{doliwa@matman.uwm.edu.pl}
\urladdr{http://wmii.uwm.edu.pl/~doliwa/}
\curraddr{Faculty of Mathematics and Computer Science, University of Warmia and Mazury in Olsztyn, ul.~S{\l}oneczna~54, 10-710 Olsztyn, Poland}

\author[S. M. Sergeev]{Sergey M. Sergeev}
\address{S. M. Sergeev: Faculty of Information Sciences and Engineering, University of Canberra, Canberra ACT 2601, Australia}
\email{Sergey.Sergeev@canberra.edu.au}

\date{}

\begin{abstract}
We define a map $S:\DD^2\times \DD^2 \dashrightarrow \DD^2\times \DD^2$, where $\DD$ is an arbitrary division ring (skew field), associated with the Veblen configuration, and we show that such a map provides solutions to the functional dynamical pentagon equation. We explain that fact in elementary geometric terms using the symmetry of the Veblen and Desargues configurations. We introduce also another map of a geometric origin with the pentagon property. We show equivalence of these maps with recently introduced Desargues maps which provide geometric interpretation to a non-commutative version of Hirota's discrete Kadomtsev--Petviashvili equation. 
Finally we demonstrate that in an appropriate gauge the (commutative version of the) maps preserves a natural Poisson structure -- the quasiclassical limit of the Weyl commutation relations. The corresponding quantum reduction is then studied. In particular, we discuss uniqueness of the Weyl relations for the ultra-local reduction of the map. We give then the corresponding solution of the quantum pentagon equation in terms of the non-compact quantum dilogarithm function. 
\end{abstract}
\maketitle

\section{Introduction}
Let $\mathcal{A}$ be an associative unital algebra over a field $\Bbbk$, an element $\boldsymbol{S}\in \mathcal{A}\otimes\mathcal{A}$ is said to satisfy the quantum pentagon equation if
\begin{equation} \label{eq:pentagon-A}
\boldsymbol{S}_{23} \boldsymbol{S}_{13} \boldsymbol{S}_{12} = \boldsymbol{S}_{12} \boldsymbol{S}_{23} \qquad \text{in} \quad \mathcal{A}\otimes \mathcal{A} \otimes \mathcal{A}, 
\end{equation}
where $\boldsymbol{S}_{ij}$ acts as $\boldsymbol{S}$ on $i$-th and $j$-th factors in the tensor product and leaves unchanged elements in the remaining factor. Equation \eqref{eq:pentagon-A} looks like a degenerate version of the quantum Yang--Baxter equation, well known in the theory of exactly solvable models of statistical mechanics and quantum field theory \cite{Baxter,QISM,JimboMiwa}, and quantum groups \cite{Kassel}. There are also some similarities in constructing solutions of both equations, for example the role of the Drinfeld double construction \cite{Drinfeld-qg} of solutions of the quantum Yang--Baxter equation is replaced by the Heisenberg double \cite{STS,DaeleKeer,Lu,Kashaev-P}. However, in modern theory of quantum groups \cite{BaajSkandalis,Woronowicz-P,KustermansVaes} (see also \cite{Timmermann} for a review, and \cite{Militaru} for discussion of the finite dimensional case) the quantum pentagon equation seems to play more profound role. Remarkably, given a solution of \eqref{eq:pentagon-A} satisfying some additional non-degeneracy conditions, it allows to construct all the remaining structure maps of a quantum group and of its Pontrjagin dual simultaneously. 

Let $\XX$ be a set, $S:\XX\times \XX \to \XX\times \XX$ be a map from its square into itself. We call $S$ \emph{pentagon map} if it satisfies the functional (or set-theoretical) pentagon equation \cite{Zakrzewski}
\begin{equation} \label{eq:pentagon-S}
S_{12}\circ  S_{13} \circ S_{23} = S_{23} \circ S_{12}, \qquad \text{on} \quad \XX \times \XX \times \XX,
\end{equation}
regarded as an equality of composite maps; here again $S_{ij}$ acts as $S$ in $i$-th and $j$-th factors of the Cartesian product. One can consider a parameter (it may be functional) dependent version, then the parameters of the five maps in \eqref{eq:pentagon-S} may be constrained by relations involving also the dynamical variables. 

The corresponding functional Yang--Baxter equation \cite{Sklyanin,Drinfeld-fYB} 
has been studied recently \cite{ABS-consistency,Veselov-YB,PapTonVes} in connection to integrability (understood as the multidimensional consistency \cite{Nijhoff-c,BS-c1,BS-nc,ABS-consistency}) of two dimensional lattice equations.
A generalization of the quantum Yang-Baxter equation for three dimensional models is the tetrahedron equation proposed by Zamolodchikov \cite{Zamolodchikov}. 
It is related \cite{BaMaSe} to four dimensional consistency of the discrete Darboux equations \cite{BoKo}, or equivalently, to four dimensional compatibility of the geometric construction scheme of the quadrilateral lattice \cite{MQL}. It turns out \cite{BaMaSe-P} that all presently known solutions of the quantum tetrahedron equation can be obtained, by a canonical quantization, from classical solutions of the functional tetrahedron equation derived from the quadrilateral lattices.

Recently it has been observed \cite{Dol-Des} that the quadrilateral lattice theory can be considered as a part of the theory of the Desargues maps, which describe in geometric terms Hirota's discrete Kadomtsev--Petviashvili (KP) equation \cite{Hirota} and its integrability properties. 
It is known that the Hirota equation encodes \cite{Miwa} the KP hierarchy of integrable equations~\cite{DKJM}. It plays also an important role in many branches of mathematics and theoretical physics related to integrability; see~\cite{KNS-rev} for a recent review of some of its application. There is a natural question what should replace Zamolodchikov's tetrahedron equation in the transition from quadrilateral lattices do Desargues maps. We demonstrate that the answer is provided by the pentagon equation.

To make presentation more transparent we start in Section~\ref{sec:geom} from the geometric essence of the construction based on elementary considerations on the Veblen and Desargues configurations \cite{BeukenhoutCameron-H} and their symmetry groups. Given five points of the Veblen configuration we are (almost) uniquely given the last one, what we call the \emph{Veblen flip}. The presence of five Veblen configurations within the Desargues configuration gives the pentagon property of the Veblen flip. 
By parametrizing the Veblen flip using homogeneous coordinates from a division ring $\DD$ we obtain in Section~\ref{sec:alg} a birational map of $\DD^2\times\DD^2$ into itself, which contains a functional gauge parameter. We check that the Veblen map satisfies the functional pentagon equation, provided the functional parameters are restricted by some additional relations. In Section~\ref{sec:norm} we discuss another solution, without parameters, of the pentagon map together with its geometric interpretation. Section~\ref{sec:Hir} is devoted to presentation of the equivalence to the above solutions of the pentagon equation with the Desargues maps and the non-commutative Hirota (discrete KP) equation. 

The quantization procedure, which we present in Section~\ref{sec:quant}, can be understood in two ways. First, under appropriate choice of the gauge in the commutative case the Veblen map leaves invariant a natural Poisson structure. Such a map preserves also the ultra-local Weyl commutation relations which quantize that structure. From other point of view the procedure can be considered as integrable reduction of the generic division ring solution of the functional pentagon equation to a particular division algebra. This allows to rewrite the map as inner automorphism, and to obtain the corresponding solution of the dynamical quantum pentagon equation, which can be expressed using the quantum dilogarithm function. Finally, in addition to the concluding remarks we present several open problems and research directions.

\section{The Veblen flip and its pentagon property}
\label{sec:geom}
Below we present elementary considerations, which form however a geometric core of the paper with far reaching consequences.
\begin{figure}
\begin{center}
\includegraphics[width=14cm]{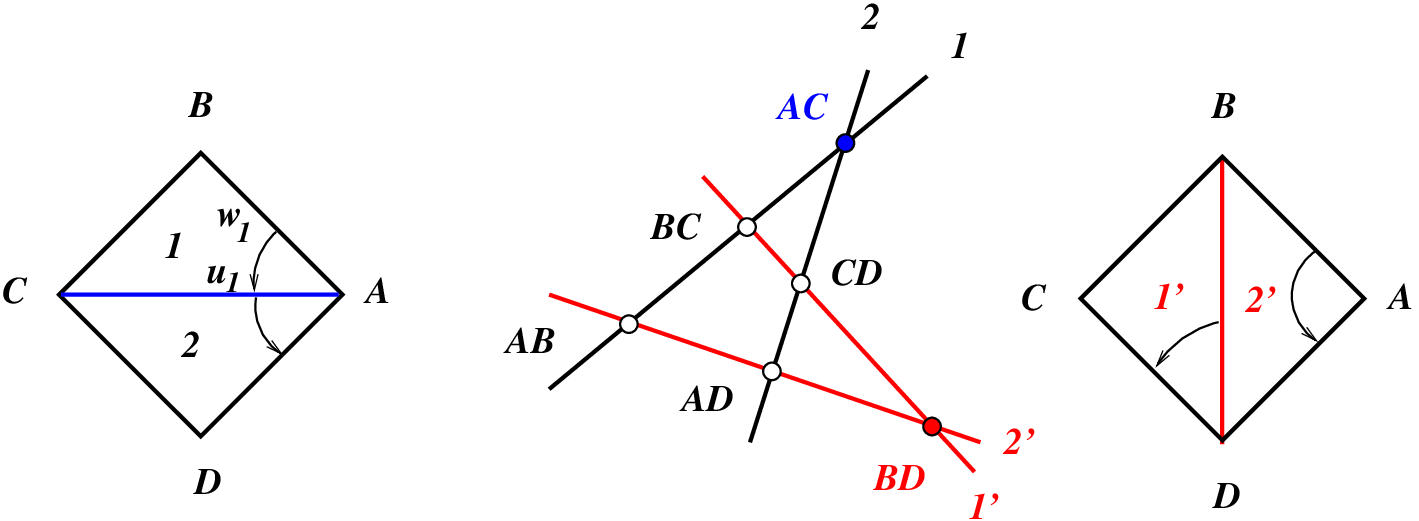}
\end{center}
\caption{The Veblen flip $f_{ABCD}$ and its tetrahedron representation. Faces of the tetrahedron represent lines of the configuration, and the arrows denote position of the $w$ and $u$ coefficients in the normalization of the corresponding linear relations discussed in Section~\ref{sec:alg}.}
\label{fig:Veblen-flip-tetr}
\end{figure}

\subsection{Geometry and combinatorics of the Veblen and Desargues configurations}
\label{sec:GC-VD}
Consider the Veblen configuration $(6_2 , 4_3)$ of six points and four lines, each point/line is incident with exactly two/three lines/points; see Fig.~\ref{fig:Veblen-flip-tetr}. 
To elucidate the $S_4$ symmetry group of the configuration is convenient to label its points by two-element subsets of the four-element set $\{A,B,C,D\}$, and lines by three-element subsets; the incidence relation is defined by containment. This combinatorial description  of the Veblen configuration has a geometric origin~\cite{Cayley}, which associates (in a non-unique way) with the configuration four points $A, B, C, D$ in general position in $\PP^3$. Six lines of edges of the simplex and its four planes intersected by a generic plane form six points and four lines the Veblen configuration (on the intersection plane). 

The Desargues configuration is of the type $(10_3)$, i.e. it consists of ten points and ten lines, each line/point is incident with exactly three points/lines. It is known that combinatorially there are ten such distinct configurations \cite{Grunbaum}. The Desargues configuration is selected by the property that it contains five Veblen configurations. Again, it is possible to label points of the configuration by two-element subsets of the five-element set $\{A,B,C,D,E\}$, and lines by three-element subsets; the incidence relation is defined by containment. Similarly, like for the Veblen configuration, given five points in general position in $\PP^4$, consider lines joining pairs of points, and planes defined by the triples. A section of such a system of ten lines and  ten planes by a generic three-dimensional hyperplane gives a Desargues configuration, see Fig.~\ref{fig:Desargues-simplex}. Five four-element subsets give rise to five Veblen configurations.

\subsection{Geometry and combinatorics of the Veblen flip}

\begin{Def}
Given two ordered pairs $( P_1, P_2 )$, $ ( P_3, P_4 )$ of distinct and non-collinear points of a projective space such that the lines $\langle P_1, P_2 \rangle$ and $ \langle P_3, P_4 \rangle$ are coplanar, and thus intersect in the point $P_5$. We denote such five points by $\{( P_1, P_2 ), ( P_3, P_4 ), P_5 \}$ and say that they satisfy the Veblen configuration condition.
\end{Def} 
We can complete these five points (and two lines), by the intersection point $P_6$ of the two lines $\langle P_1, P_3 \rangle$ and $ \langle P_2, P_4 \rangle$ (and the lines), to the Veblen configuration; notice that the ordering in pairs is important, because $P_6\neq\langle P_1, P_4 \rangle \cap \langle P_2, P_3 \rangle$. If we remove the old intersection point $P_5$, which does not belong to two new intersecting lines we obtain new system $\{( P_1, P_3 ), ( P_2, P_4 ), P_6 \}$ satisfying the Veblen configuration condition. Such an involutory transition we call \emph{a Veblen flip}. 

Given five points satisfying the Veblen configuration condition we can label them, in the way described above in Section~\ref{sec:GC-VD}, by five edges of the three-simplex or, equivalently (up to an action of $S_4$), by five two-element subsets of a four element set. The edge representing the intersection point $P_5$ contains two vertices of valence three, the sixth edge will represent the point $P_6$; see Fig.~\ref{fig:Veblen-flip-tetr} which illustrates the Veblen flip $\{( AB, BC ), ( AD, CD ), AC \} \mapsto \{( AB, AD ), ( BC, CD ), BD \} $ . 

\subsection{Geometry of the pentagon relation satisfied by the Veblen flip}
\begin{figure}
\begin{center}
\includegraphics[width=12cm]{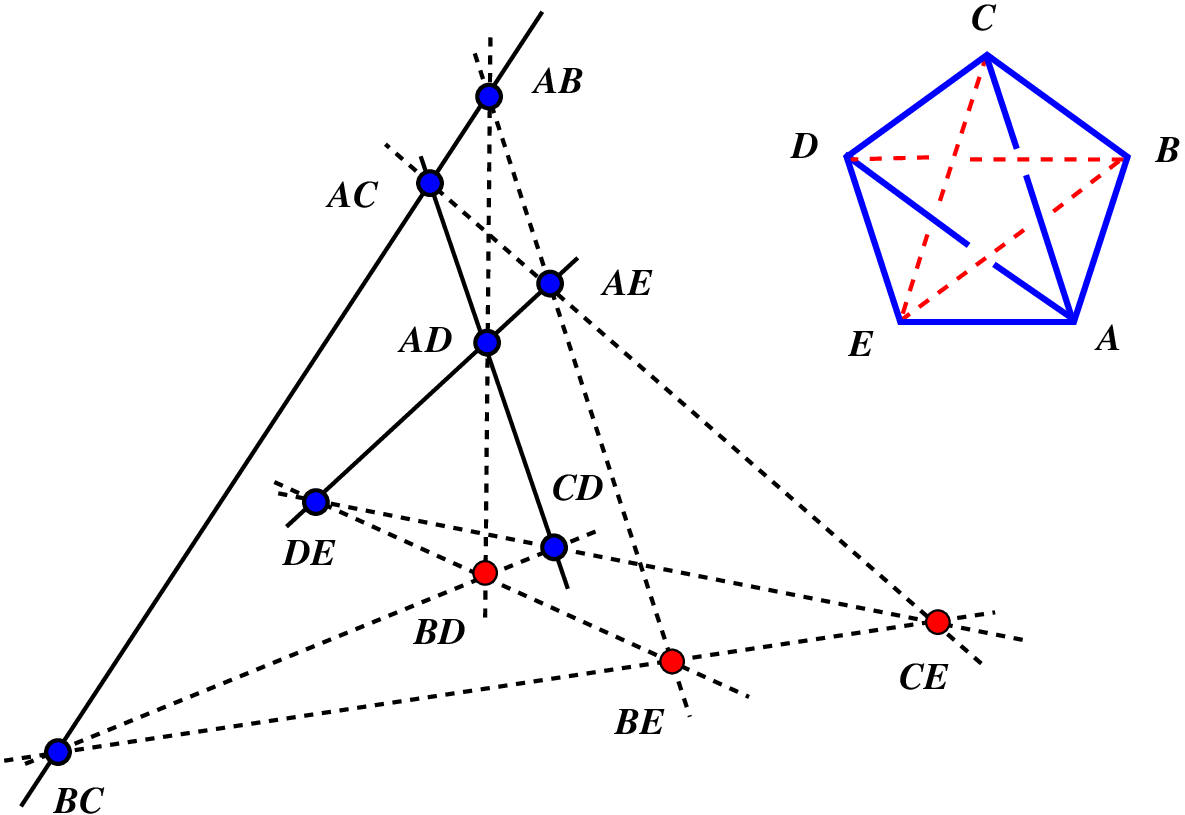}
\end{center}
\caption{The Desargues configuration and its four-simplex combinatorics. The (initial) seven points and three lines, and their four-simplex counterparts, used to study the pentagon property of the Veblen flip are  distinguished by solid lines.}
\label{fig:Desargues-simplex}
\end{figure}
To understand the geometric origin of the pentagon relation property of the Veblen flip let us start from two sets of points satisfying the Veblen configuration property with three points and one line in common, see Fig.~\ref{fig:Desargues-simplex}. Without loss of generality (the symmetry group of Desargues configuration is the permutation group $S_5$) we consider points  $AC$, $AD$, $CD$ of the line $ACD$, two points $AB$ and $BC$ on the line $ABC$, and two points $AE$ and $DE$ on the line $ADE$. By a sequence of Veblen flips we can recover all the other points of the Desargues configuration.
\begin{Rem}
The Desargues configuration is considered usually in relation to the celebrated Desargues theorem valid in projective spaces over division rings~\cite{BeukenhoutCameron-H}. It states that two triangles (e.g. $\triangle_{AB,AC,AE}$ and $\triangle_{BD,CD,DE}$ on Fig.~\ref{fig:Desargues-simplex}) are in perspective from a point ($AD$ in our case) if and only if they are in perspective from a line (that passing through the remaining three points $BC$, $CE$, $BE$ of the configuration, which are constructed as intersections of the corresponding sides of the triangles). 
\end{Rem}
If we want to apply the Veblen flip, starting from the initial configuration described above, it can be either the flip $f_{ABCD}$ (in the configuration labelled by the tetrahedron with vertices $A,B,C,D$)
or the flip $f_{ACDE}$. After the first flip we have a similar choice of two flips, but we exclude that related to the Veblen configuration just used (in order not to go back immediately to the initial configuration). Therefore the choice of the first flip determines the next ones. By inspection of Fig.~\ref{fig:pentagon-f} we obtain that  superposition $f_{BCDE}\circ f_{ABCE} \circ f_{ACDE}$ of transformations when applied to the initial configuration, gives the same result as $f_{ABDE} \circ f_{ABCD}$.
\begin{Rem}
In other words, starting from the data described above and going around the diagram we obtain the sequence of Veblen flips of period five.
\end{Rem}
\begin{figure}
\begin{center}
\includegraphics[width=14cm]{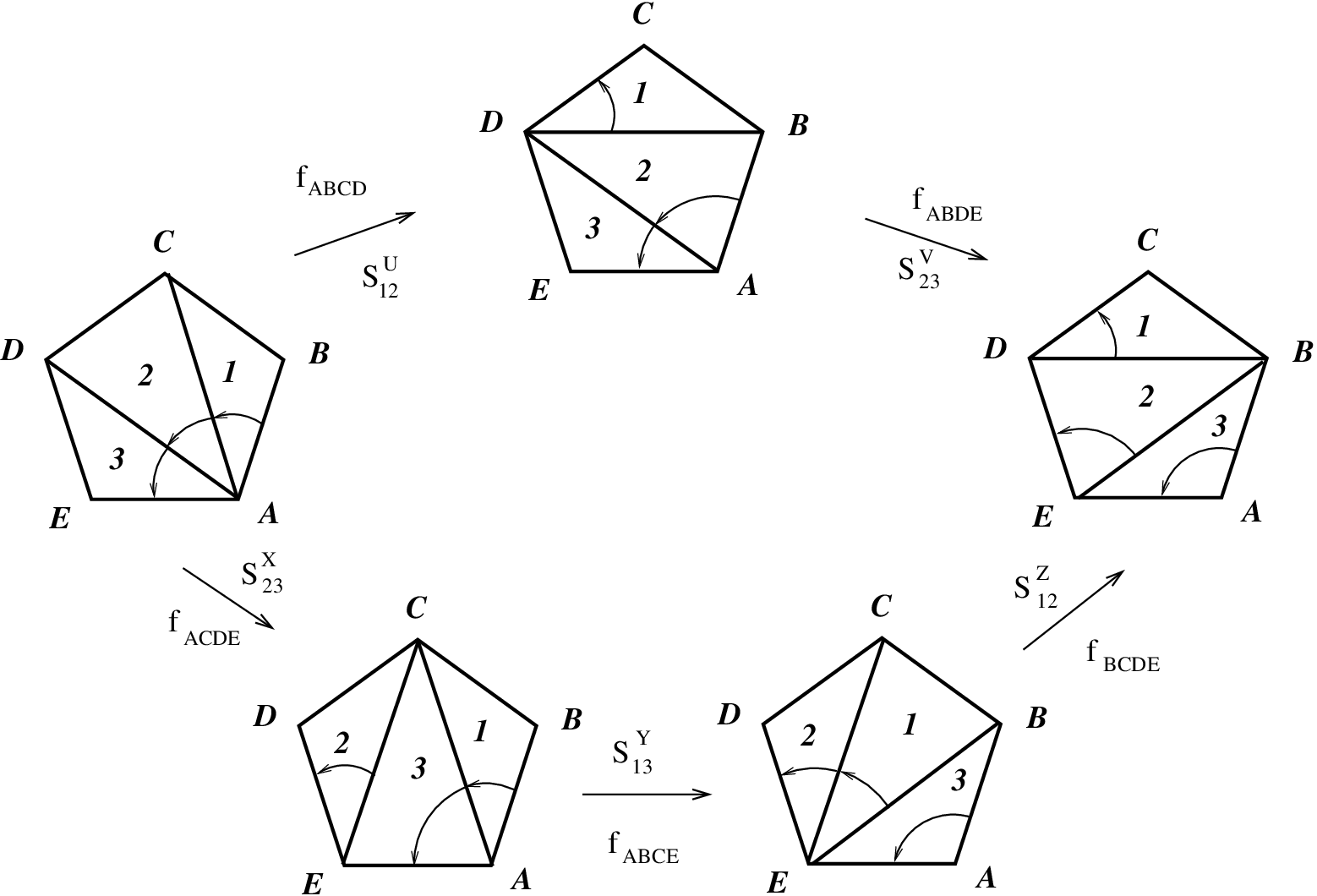}
\end{center}
\caption{The pentagon relation for Veblen flips in the 4-simplex representation}
\label{fig:pentagon-f}
\end{figure}

\section{Algebraic description of the Veblen flip and of its pentagon property}
\label{sec:alg}
The present section is devoted to algebraization of the geometric considerations presented above. In doing that we introduce numerical coefficients (in the non-commutative Hirota equation interpretation they will serve as soliton fields, see Section~\ref{sec:Hir}) attached to vertices of the Veblen configuration. These describe positions of points on the corresponding lines, and the Veblen flip gives a map from a set of old coefficients to the corresponding set of new ones. The pentagon property of the Veblen flip implies that the map satisfies the functional dynamical pentagon equation.

\subsection{Algebraic description of the Veblen flip} 
\label{sec:V-pent}
Consider three distinct collinear points $AB$, $AC$, $BC$ of the (right) projective space $\PP^M(\DD)$ over division ring $\DD$. We label them in the spirit of the combinatorics of the Veblen and Desargues configurations, by edges of a triangle $ABC$. Their collinearity, expressed in terms of the homogeneous coordinates 
$\bphi_{AB}$, $\bphi_{AC}$, $\bphi_{BC} \in \DD^{M+1}$, takes the form of a linear dependence relation 
\begin{equation} \label{eq:collinarity-AB}
\bphi_{BC} = \bphi_{AB} w - \bphi_{AC} u,
\end{equation}
with non-vanishing coefficients $u,w\in{\DD}^{\times}$ (we add the minus sign for convenience).
The above linear dependence relation \eqref{eq:collinarity-AB} is normalized at $\phi_{BC}$ by putting corresponding coefficient equal to one. We can incorporate this new information to the graphic representation on the triangle $ABC$ by adding an arrow from the $w$-edge (i.e. the edge which represents the point with homogeneous coordinates multiplied by the coefficient $w$) to the $u$-edge, as visualized on Fig.~\ref{fig:Veblen-flip-tetr}. 

Consider five distinct points of projective space $\PP^M(\DD)$, which satisfy the Veblen configuration condition and are labeled by five edges of a three-simplex, as explained in Section~\ref{sec:geom}. Two triplets of collinear points give rise to two linear dependence relations. In order to write them down we have to fix enumeration of the lines, and to normalize the corresponding linear relations. 

It turns out that in discussing the pentagon property of the Veblen map it is convenient to choose the normalization in relation to labelling of straight lines as follows: 
\begin{enumerate}
\item Pick up a point ($BC$ on Fig.~\ref{fig:Veblen-flip-tetr}) different from the initial intersection point ($AC$ on Fig.~\ref{fig:Veblen-flip-tetr}) and declare it to belong to the lines $1$ and $1^\prime$, moreover the point is the normalization point on the both lines $1$ and $1^\prime$.
\item The second point on line $1$ ($AB$ on Fig.~\ref{fig:Veblen-flip-tetr}) different from the intersection point is attributed the $w$ coefficients on both lines $1$ and $2^\prime$.
\item The old intersection point is attributed the $w$ coefficient on line $2$, and the new intersection point ($BD$ on Fig.~\ref{fig:Veblen-flip-tetr}) is the $w$-point on line $1^\prime$.
\item The point on the lines $2$ and $2^\prime$ ($AD$ on Fig.~\ref{fig:Veblen-flip-tetr}) is attributed the $u$ coefficients on both the lines.
\end{enumerate}

In algebraic terms we have, as the starting point, homogeneous coordinates of five initial points and four coefficients $w_1$, $u_1$, $w_2$, $u_2$, in the linear relations
\begin{align}
\bphi_{BC} = & \bphi_{AB}w_1 - \bphi_{AC} u_1,\\
\bphi_{CD} = & \bphi_{AC}w_2 - \bphi_{AD} u_2.
\end{align}
In writing down similar linear dependence relations for coordinates of points on two new lines of the Veblen configuration
\begin{align}
\bphi_{BC} = & \bphi_{BD}w_1^\prime - \bphi_{CD} u_1^\prime , \label{eq:V-uw1} \\
\bphi_{BD} = & \bphi_{AB} w_2^\prime - \bphi_{AD} u_2^\prime, \label{eq:V-uw2}
\end{align}
we are looking for coordinates $\bphi_{BD}$ of the new intersection point, and for the corresponding coefficients $w_1^\prime$, $u_1^\prime$, $w_2^\prime$, $u_2^\prime$. The transformation formulas read
\begin{align} \label{eq:S-wu1}
w_1^\prime & = G u_1,  & u_1^\prime & = w_2^{-1} u_1 , \\
\label{eq:S-wu2}
w_2^\prime & = w_1 u_1^{-1} G^{-1}, & u_2^\prime  & = u_2 w_2^{-1} G^{-1}.
\end{align}
Here $G$ is a free parameter which expresses possibility of multiplying $\bphi_{BD}$ by a non-zero factor. Indeed, equations \eqref{eq:V-uw1}, \eqref{eq:V-uw2} give
\begin{equation} \label{eq:S-G-BD}
\bphi_{BD} = 
\left( \bphi_{BC} u_1^{-1} + \bphi_{CD} w_2^{-1} \right) G^{-1} = \left( \bphi_{AB} w_1 u_1^{-1} - \bphi_{AD} u_2w_2^{-1}\right) G^{-1}.
\end{equation}
\begin{Rem}
The inverse transformation is given as follows
\begin{align} \label{eq:S-inv-wu1}
w_1 & = w_2^\prime w_1^\prime, & u_1 & = G^{-1} w_1^\prime ,\\ 
\label{eq:S-inv-wu2}
w_2 & = G^{-1} w_1^\prime u_1^{\prime \, -1} , & 
u_2 & = u_2^\prime w_1^\prime u_1^{\prime \, -1}.
\end{align}
Actually, we could take an arbitrary non-zero parameter $G^\prime$, instead of $G$, which would result in multiplying of the original homogeneous coordinates $\bphi_{AC}$ by an appropriate factor.
\end{Rem}

Equations \eqref{eq:V-uw1} and \eqref{eq:V-uw2} have initial interpretation on \emph{the linear algebra level} as transformation (involving a numerical parameter) between numerical coefficients describing positions of points of the Veblen configuration. There is another level of looking on the formulas, which can be called \emph{the algebraic geometry level}, i.e. we consider $w_1$, $u_1$, $w_2$, and $u_2$ as non-commuting indeterminates in the corresponding universal skew field of fractions~\cite{Cohn-constructions}.  Then also the gauge coefficient $G$ can vary when changing the indeterminates, what allows for interpretation of $G=G(w_1,u_1,w_2,u_2)$ as an arbitrary rational function of the four variables. In this way by choosing function $G$ we may interpret formulas \eqref{eq:V-uw1}-\eqref{eq:V-uw2} as a definition of the~rational map $S^G:\DD^2\times \DD^2 \dashrightarrow \DD^2\times \DD^2$. Demanding invertibility of the map $S^G$ we restrict from now on our attention to birational maps, what imposes certain conditions on the form of admissible gauge functions $G$. 
\begin{Rem}
On the algebraic geometry level
in the case of birational map $S^G$, we are looking for the inverse map in the form of equations 
\eqref{eq:S-inv-wu1}-\eqref{eq:S-inv-wu2}. Therefore we can interpret the gauge parameter $G^\prime$ as a new rational function 
$G^\prime = G^\prime (w_1^\prime, u_1^\prime, w_2^\prime, u_2^\prime)$ of four indeterminates. Equality of the gauge parameters on the linear algebra level is then transfered into the following  functional relation between the gauge functions, which we write down in its full expanded form
\begin{equation} \label{eq:G-prime}
\begin{split}
G^\prime & \left( G(w_1, u_1, w_2, u_2) u_1, w_2^{-1}u_1, 
w_1 u_1^{-1} G(w_1, u_1, w_2, u_2)^{-1} , u_2 w_2^{-1} G(w_1, u_1, w_2, u_2)^{-1} \right)\\ & =  G(w_1, u_1, w_2, u_2). 
\end{split}
\end{equation}
In short notation, the functional relation \eqref{eq:G-prime} reads
\begin{equation*}
G^\prime (w_1^\prime, u_1^\prime, w_2^\prime, u_2^\prime) = G(w_1, u_1, w_2, u_2),
\end{equation*}
where $w_1^\prime$, $u_1^\prime$, $w_2^\prime$, $u_2^\prime$ are given by equations \eqref{eq:S-wu1}-\eqref{eq:S-wu2}.
\end{Rem}

\begin{Ex} \label{ex:gauge}
Fix the gauge function $G$ by requiring $w_2^\prime = w_2$, which gives 
\begin{equation*}
G(w_1,u_1,w_2,u_2) = w_2^{-1} w_1 u_1^{-1},
\end{equation*}
and the transformation formulas read
\begin{align*} 
w_1^\prime & = w_2^{-1} w_1,  & u_1^\prime & = w_2^{-1} u_1 , \\
w_2^\prime & = w_2, & u_2^\prime  & = u_2 w_2^{-1} u_1 w_1^{-1} w_2.
\end{align*}
The inverse transformation is of the form
\begin{align*} 
w_1 & = w_2^\prime w_1^\prime, & u_1 & = w_2^\prime u_1^\prime ,\\ 
w_2 & = w_2^\prime , & 
u_2 & = u_2^\prime w_1^\prime u_1^{\prime \, -1},
\end{align*}
and defines the function
\begin{equation*}
G^\prime (w_1^\prime, u_1^\prime, w_2^\prime, u_2^\prime) = w_1^\prime u_1^{\prime \, -1} w_2^{\prime\, -1}
\end{equation*}
which satisfies condition \eqref{eq:G-prime}. 
\end{Ex}

\subsection{Algebraic description of the pentagon relation}
The algebraic counterpart of the pentagon relation satisfied by the Veblen flips $f$ will be an analogous relation on the algebraic geometry level of the map $S^G$. 

Let us first discuss the implication of the pentagon relation between the Veblen flips on the linear algebra level. 
Starting from the initial configuration of seven points on three lines 
(see Fig.~\ref{fig:Desargues-simplex}) we have three linear relations
\begin{align*}
\bphi_{BC} = & \bphi_{AB}w_1 - \bphi_{AC} u_1,\\
\bphi_{CD} = & \bphi_{AC}w_2 - \bphi_{AD} u_2,\\
\bphi_{DE} = & \bphi_{AD}w_3 - \bphi_{AE} u_3,
\end{align*}
normalized according to Fig.~\ref{fig:pentagon-f}. We perform the transformation in the first two equations with the gauge parameter $U$ \, and then we perform the transformation in the second and the third equation with the gauge parameter $V$. At the end we obtain the final linear relations
\begin{align*}
\bphi_{BC} = & \bphi_{BD}\hat{w}_1 - \bphi_{CD} \hat{u}_1,\\
\bphi_{BD} = & \bphi_{BE} \hat{w}_2 - \bphi_{DE} \hat{u}_2,\\
\bphi_{BE} = & \bphi_{AB} \hat{w}_3 - \bphi_{AE} \hat{u}_3,
\end{align*}
compare Fig.~\ref{fig:pentagon-f}. 

The resulting transformation $(w_i,u_i)_{i=1}^3 \to (\hat{w}_i,\hat{u}_i)_{i=1}^3$ of the coefficients reads (abusing the notation we write $S^U_{12}$ and $S^V_{23}$ like on the algebraic geometry level)
\begin{equation*}
\left( \begin{array}{ll}
w_1 & u_1 \\
w_2 & u_2 \\
w_3 & u_3
\end{array} \right)
\xrightarrow{S_{12}^U}
\left( \begin{array}{ll}
U u_1 & w_2^{-1} u_1 \\
w_1 u_1^{-1} U^{-1} & u_2 w_2^{-1} U^{-1} \\
w_3 & u_3
\end{array} \right)
\xrightarrow{S_{23}^V}
\left( \begin{array}{ll}
U u_1 & w_2^{-1} u_1 \\
V u_2 w_2^{-1} U^{-1} & w_3^{-1} u_2 w_2^{-1} U^{-1} \\
w_1 u_1^{-1} w_2 u_2^{-1} V^{-1}  & u_3 w_3^{-1} V^{-1}
\end{array} \right).
\end{equation*}

According to our previous geometric considerations the final set of linear relations can be also obtained by the sequence of three transformations: (i) in the second and the third relation with a parameter which we call $X$, (ii) in the first and third (new) relation with a parameter $Y$, (iii) in the first and the second relation with a parameter $Z$
\begin{gather*}
\left( \begin{array}{ll}
w_1 & u_1 \\
w_2 & u_2 \\
w_3 & u_3
\end{array} \right)
\xrightarrow{S_{23}^X}
\left( \begin{array}{ll}
w_1 & u_1 \\
X u_2 &  w_3^{-1} u_2 \\
w_2 u_2^{-1} X^{-1} & u_3 w_3^{-1} X^{-1}
\end{array} \right)
\xrightarrow{S_{13}^Y}
\left( \begin{array}{ll}
Y u_1 & X u_2 w_2^{-1} u_1 \\
X u_2 &  w_3^{-1} u_2 \\
w_1 u_1^{-1} Y^{-1}  & u_3 w_3^{-1} u_2 w_2^{-1} Y^{-1}
\end{array} \right) \xrightarrow{S_{12}^Z}\\
\xrightarrow{S_{12}^Z}
\left( \begin{array}{ll}
ZX u_2 w_2^{-1} u_1 & w_2^{-1} u_1 \\
Y w_2 u_2^{-1} X^{-1} Z^{-1} & w_3^{-1} X^{-1} Z^{-1} \\
w_1 u_1^{-1} Y^{-1}  & u_3 w_3^{-1} u_2 w_2^{-1} Y^{-1}
\end{array} \right).
\end{gather*}
Notice however, that the actual equality holds only on the geometric level of points in the projective space, while on the (linear algebra) level of their homogeneous coordinates we should take into account possibile rescaling. Homogeneous coordinates of the two new points $BD$ and $BE$ (the second sequence of transformations gives also the point $CE$ which is not produced by the first sequence) have then double expressions 
\begin{align*}
\bphi_{BD}  & = \left(  \bphi_{BC} u_1^{-1} + \bphi_{CD} w_2^{-1} \right) U^{-1} =
\left(  \bphi_{BC} u_1^{-1} w_2 + \bphi_{CD} \right) u_2^{-1} X^{-1} Z^{-1},\\
\bphi_{BE} & = \left( \bphi_{AB} w_1 u_1^{-1} w_2 u_2^{-1} - \bphi_{AE} u_3 w_3^{-1}  \right) V^{-1} =  \left( \bphi_{AB} w_1 u_1^{-1}  - \bphi_{AE} u_3 w_3^{-1} u_2 w_2^{-1} \right) Y^{-1},
\end{align*}
found with the help of equation \eqref{eq:S-G-BD}. Therefore, to have equality of the homogenous coordinates of the new points and of the coefficients in the both expressions of the final linear relations the gauge parameters should be adjusted according to the following equations
\begin{equation} \label{eq:UV-XYZ}
U = ZX u_2 w_2^{-1}, \qquad
V = Y w_2 u_2 ^{-1}.
\end{equation}

On the algebraic geometry level equations \eqref{eq:UV-XYZ} have more complicated interpretation. Consider five rational functions $U,V,X,Y,Z\colon \DD^2 \times \DD^2 \dashrightarrow \DD$ of four variables. 
We call the functions pentagon-compatible if they satisfy in  $\DD^2 \times \DD^2\times\DD^2 $ equations of the form \eqref{eq:UV-XYZ}, where we write (the order of functions in the system is important)
\begin{gather}
\nonumber U = U(w_1,u_1,w_2,u_2), \qquad X = X(w_2,u_2,w_3,u_3), \\ \label{eq:UV-XYZ-funct}
V = V(w_1 u_1^{-1} U^{-1},u_2 w_2^{-1} U^{-1}, w_3,u_3), \qquad 
Y = Y(w_1,u_1, w_2 u_2^{-1} X^{-1}, u_3 w_3^{-1} X^{-1}), \\
\nonumber Z = Z(Y u_1 , X u_2 w_2^{-1} u_1 , X u_2 , w_3^{-1} u_2 ).
\end{gather}
Notice that the arguments of the functions above coincide with arguments of the corresponding Veblen maps considered in the two sequences of transformations. In the expanded form the functional equations look rather complicated, and one can even wonder if there exists any pentagon-compatible system of functions.
\begin{Ex} \label{ex:gauge-2}
If the gauge functions $U,V,X,Y,Z$ are of the same form as $G$ in Example~\ref{ex:gauge} then this choice gives a solution to equations \eqref{eq:UV-XYZ}-\eqref{eq:UV-XYZ-funct}.
\end{Ex}

The Theorem below can be verified directly, but actually it follows from the considerations above.
\begin{Th} \label{th:pentagon-S}
Given five rational functions of four variables $U,V,X,Y,Z\colon \DD^2 \times \DD^2 \dashrightarrow \DD$ which are pentagon-compatible then the Veblen maps $S^{G}:\DD^2\times\DD^2 \dashrightarrow \DD^2\times\DD^2 $, where $G$ is one of $U,V,X,Y,Z$, satisfy the functional dynamical pentagon relation on $\DD^2\times\DD^2 \times\DD^2$
\begin{equation} \label{eq:pentagon-funct-dynam}
S_{12}^Z \circ S_{13}^Y \circ S_{23}^X = S_{23}^V \circ S_{12}^U.
\end{equation}
\end{Th}

\section{The normalization map and its pentagon property}
\label{sec:norm}
In this Section we present another map with the pentagon property. We give also its simple geometric meaning related to four collinear points. This map will be used in the next Section, where we discuss its role in establishing the relation of the Veblen map to the Hirota equation. First we discuss another auxiliary map of period three.

Consider a map of vertices $A,B,C$ of a triangle into a projective space $\PP^M(\DD)$ such that all the vertices are mapped into collinear (but different) points. The corresponding linear constraint involving the homogeneous coordinates of the points
\begin{equation*}
\bphi_A = \bphi_C y - \bphi_B x
\end{equation*}
we call normalized at $\bphi_A$, while $\bphi_B$ and $\bphi_C$ are called coordinates of the $x$ and $y$-point, respectively. Graphically, we put an arrow at the vertex of the triangle representing the normalization point. The arrow is directed from the edge joining the (vertex representing the) normalization point and the $x$-point, see Fig.~\ref{fig:triangle-lin-ABCD}. Notice that this is the different normalization rule than that used in Section~\ref{sec:V-pent}. 
\begin{figure}
\begin{center}
\includegraphics[width=14cm]{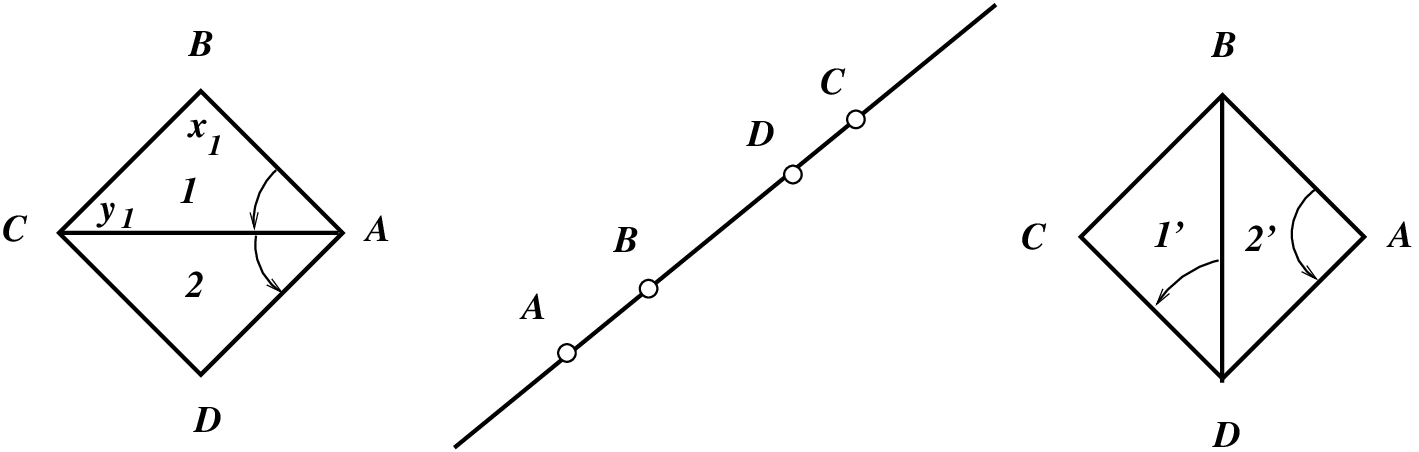}
\end{center}
\caption{Graphic representation of the linear relations for the normalization map}
\label{fig:triangle-lin-ABCD}
\end{figure}
\begin{Rem}
Notice that by changing the normalization point "by $2\pi/3$ rotation", where the orientation is given by the arrow, we arrive at an equivalent linear relation 
\begin{equation}
\bphi_{B} = \bphi_{A} \tilde{y} - \bphi_{C} \tilde{x},
\end{equation}
where 
\begin{equation} \label{eq:Newton-D}
\left( \tilde{y}, \tilde{x} \right) = N(y,x) = \left( - x^{-1}, - y x^{-1} \right), \qquad (y,x ) = ( \tilde{x} \tilde{y}^{-1}, - \tilde{y}^{-1}), \qquad N^3 = \mathrm{id}. 
\end{equation}
The birational map $N:\DD^2 \dashrightarrow \DD^2$, which can be called the non-commutative Newton map~\cite{Dieudonne}, is of order three.
\end{Rem}

Let us consider four collinear (but distinct) points of a projective space $\PP^M(\DD)$. We label them by vertices $A,B,C,D$ of a three-simplex. Four triangular faces of the simplex give four linear relations (up to possible normalizations). By analogy with the Veblen map presented in Section~\ref{sec:V-pent} we write first two of them
\begin{align} \label{eq:lin-S-bar-12}
\bphi_A = & \bphi_C y_1 - \bphi_B x_1,\\
\bphi_A = & \bphi_D y_2 - \bphi_C x_2,
\end{align}
and transfer them into the remaining equations 
\begin{align}
\bphi_D = & \bphi_C \bar{y}_1 - \bphi_B \bar{x}_1,\\
\bphi_A = & \bphi_D \bar{y}_2 - \bphi_B \bar{x}_2.\label{eq:lin-S-bar-12-prime}
\end{align}
The relation between old and new coefficients reads
\begin{align} \label{eq:S-bar-yx1}
\bar{y}_1 & = (y_1 + x_2) y_2^{-1},  & \bar{x}_1 & = x_1 y_2^{-1} \\
\bar{y}_2 & =  y_2 (x_2 + y_1)^{-1}y_1, &\bar{x}_2 &= x_1 (x_2 + y_1)^{-1}x_2, \label{eq:S-bar-yx2}
\end{align}
and defines a birational map $\bar{S}:\DD^2\times\DD^2 \dashrightarrow  \DD^2\times\DD^2$.
The inverse transformation is given by
\begin{align}
y_1 & = \bar{y}_1 \bar{y}_2  , & x_1 & = \bar{x}_2 + \bar{x}_1 \bar{y}_2 \\
y_2 & =  \bar{y}_2  + \bar{x}_1^{-1}\bar{x}_2, & x_2 &= \bar{y}_1 \bar{x}_1^{-1}\bar{x}_2.
\end{align}
The map $\bar{S}:\DD^2\times\DD^2 \dashrightarrow \DD^2\times\DD^2$, given by equations \eqref{eq:S-bar-yx1}-\eqref{eq:S-bar-yx2} is called \emph{the (change of the) normalization} map.
\begin{Rem}
This time homogeneous coordinates of all points are fixed, which results in the absence of gauge parameters in the map.
\end{Rem}
\begin{Rem}
Equations \eqref{eq:S-bar-yx2} can be rewritten in the form
\begin{equation} \label{eq:S-Hir-alg}
\bar{y}_2^{-1} y_2 - y_1^{-1}x_2 = 1, \qquad \bar{x}_2^{-1} x_1 - x_2^{-1} y_1 = 1,
\end{equation}
which will be relevant in Section~\ref{sec:Hir}.
\end{Rem}

Finally, let us consider five collinear points and label them by vertices of the four-simplex. 
\begin{Prop}
The normalization map $\bar{S}:\DD^2\times\DD^2 \dashrightarrow \DD^2\times\DD^2$, given by equations \eqref{eq:S-bar-yx1}-\eqref{eq:S-bar-yx2} satisfies the functional pentagon relation \eqref{eq:pentagon-S}.
\end{Prop}
\begin{proof}
The result can be verified by direct calculation. Notice however that, even having much simpler geometric interpretation then the Veblen map $S^G$, the new map $\bar{S}$ has the same graphic representation (compare Fig.~\ref{fig:Veblen-flip-tetr} with Fig.~\ref{fig:triangle-lin-ABCD}). Therefore the statement follows from consistency of arrows on Fig.~\ref{fig:pentagon-f}.
\end{proof}

\section{Relation to Desargues maps and Hirota's discrete KP equation}
\label{sec:Hir}
In \cite{Dol-Des} it was shown (see also an earlier related work~\cite{KoSchief-Men}) that the non-commutative version of Hirota's discrete KP equation \cite{Hirota,FWN-Capel,Nimmo-NCKP} can be derived from the Veblen configuration. Moreover, its four dimensional compatibility follows from the Desargues theorem. The important observation that the four dimensional consistency of the discrete KP equation in the Schwarzian form has combinatorics of the Desargues configuration is due to Wolfgang Schief, see \cite{Bobenko-talk} and remarks in \cite{Dol-Des,Dol-AN}. In view of the previous Sections it is clear that the incidence geometry structures are the same for both the non-commutative Hirota system and the Veblen map. The goal of this Section is to establish a dictionary between both subjects. The Reader not interested in the theory of integrable discrete equations may go directly to the next part where we discuss the quantum pentagon equation.
\subsection{Desargues maps and the Hirota system}
We collect first some facts on incidence geometric interpretation of the Hirota equation. The Desargues maps, as defined in \cite{Dol-Des}, are maps 
$\phi:\ZZ^N\to\PP^M(\DD)$ of multidimensional integer lattice into 
projective space of dimension $M\geq 2$ over a division ring $\DD$, 
such that for any pair of indices 
$i\ne j$ the points $\phi(n)$, $\phi(n+\boldsymbol{\varepsilon}_i)$ and 
$\phi(n+\boldsymbol{\varepsilon}_j)$ are collinear; here 
$\boldsymbol{\varepsilon}_i = (0, \dots, \stackrel{i}{1}, \dots ,0)$
is the $i$-th element of the canonical basis of $\RR^N$. We will write $F_{(i)}(n)$ instead of $F(n+\boldsymbol{\varepsilon}_i)$ for any function $F$ on $\ZZ^N$. Moreover we will often skip the argument $n$.
\begin{Rem}
The Desargues maps can be defined starting from the root lattice $Q(A_N)$, instead of the $\ZZ^N$ lattice. Such an approach, proposed in \cite{Dol-AN}, makes transparent the $A_N$ affine Weyl group symmetry of Desargues maps and the discrete KP system. Its "local" version is the pentagon property of the Veblen map.
\end{Rem}

In the homogeneous coordinates $\bphi:\ZZ^N\to\DD^{M+1}$  the map can be
described in terms of the linear system
\begin{equation} \label{eq:lin-A}
\bphi + \bphi_{(i)} A_{ij} + \bphi_{(j)} A_{ji} = 0, \qquad i\ne j,
\end{equation}
where $A_{ij}:\ZZ^N\to {\DD}^{\times}$ are certain non-vanishing functions.
The compatibility of the linear system \eqref{eq:lin-A} turns out to be equivalent to
equations 
\begin{align} \label{eq:alg-cond}
& A_{ij}^{-1} A_{ik} + A_{kj}^{-1} A_{ki} = 1, \\
\label{eq:shift-cond}
&A_{ik(j)}A_{jk} = A_{jk(i)} A_{ik},
\end{align}
where the indices $i,j,k$ are distinct. Equations \eqref{eq:alg-cond} and \eqref{eq:shift-cond} are equivalent, in an appropriate gauge given in \cite{Dol-Des}, to the non-commutative Hirota system proposed in \cite{Nimmo-NCKP}.

The precise relation with the Hirota equation is as follows (see \cite{Dol-Des} for details and other gauge forms). Equations \eqref{eq:alg-cond} and \eqref{eq:shift-cond} imply existence of a  non-vanishing function $F:\ZZ^N \to \DD^\times$ satisfying
\begin{equation} \label{eq:gauge-H}
F_{(i)} A_{ij} = - F_{(j)} A_{ji} , \qquad i\neq j.
\end{equation}
After rescaling the homogeneous coordinates as
\begin{equation}
\tilde\bphi = \bphi F^{-1},
\end{equation}
we obtain that
$\tilde\bphi$ satisfies the linear problem \cite{DJM-II,Nimmo-NCKP}
\begin{equation} \label{eq:lin-dKP}
\tilde\bphi_{(i)} - \tilde\bphi_{(j)} =  \tilde\bphi U_{ij},  \qquad i \ne j \leq N,
\end{equation}
with the
coefficients
\begin{equation}
U_{ij} = F A_{ji}^{-1}F_{(j)}^{-1} = - U_{ji}.
\end{equation}
Moreover, equations
\eqref{eq:alg-cond}-\eqref{eq:shift-cond} reduce
to the following systems for distinct triples $i,j,k$
\begin{align} \label{eq:alg-comp-U}
& U_{ij} + U_{jk} + U_{ki} = 0, \\
& \label{eq:U-rho} 
U_{kj}U_{ki(j)} = U_{ki} U_{kj(i)}.
\end{align}
Equations \eqref{eq:U-rho} allow to introduce potentials $r_i:\ZZ^N\to{\DD}^{\times}$ such that
\begin{equation} \label{eq:def-r}
r_{i(j)} = r_i U_{ij}, \qquad i\ne j.
\end{equation}
When $\DD$ is commutative then the functions $r_i$ can 
be parametrized in terms of a 
single potential $\tau$
\begin{equation} \label{eq:r-tau}
r_i = (-1)^{\sum_{k<i}n_k}\frac{\tau_{(i)}}{\tau},
\end{equation}
and equation
\eqref{eq:alg-comp-U} reduces to the Hirota equation 
\begin{equation} \label{eq:H-M}
\tau_{(i)}\tau_{(jk)} - \tau_{(j)}\tau_{(ik)} + \tau_{(k)}\tau_{(ij)} =0,
\qquad 1\leq i< j <k \leq N.
\end{equation}

\subsection{The non-commutative Hirota system, the normalization map, and the Veblen map}
As it was discussed in \cite{Dol-Des} the three dimensional compatibility of the Desargues maps (or equivalently, the compatibility of the linear problem of the non-commutative Hirota equation) can be stated as the Veblen theorem in the form: given four distinct points $\phi_{(j)}$, $\phi_{(k)}$, $\phi_{(ij)}$, $\phi_{(ik)}$; if the lines $\langle \phi_{(j)}\phi_{(k)}\rangle = L$ and $\langle\phi_{(ij)}\phi_{(ik)}\rangle=L_{(i)}$ intersect, then the lines 
$\langle \phi_{(j)}\phi_{(ij)}\rangle = L_{(j)}$ and $\langle \phi_{(k)}\phi_{(ik)}\rangle  =L_{(k)}$ 
intersect as well, see Fig.~\ref{fig:Desargues-Veblen-3d}, in the point $\phi_{(jk)}$.
\begin{figure}
\begin{center}
\includegraphics[width=7cm]{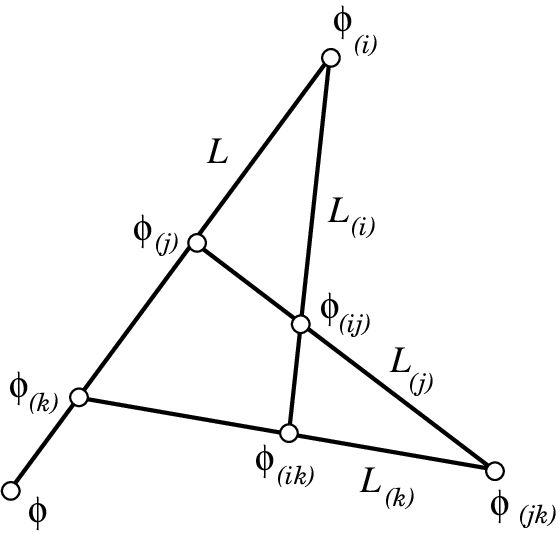}
\end{center}
\caption{Three dimensional compatibility of the Desargues map condition as the Veblen configuration}
\label{fig:Desargues-Veblen-3d}
\end{figure}
The compatibility of the linear problem is the non-commutative Hirota equation. As the geometry of both the Hirota system and the Veblen map is the same, we have only to identify the functions $A_{ij}$ entering into the equation with the coefficients $y$, $x$, $w$ and $u$. Notice the presence of the additional point $\phi$ outside of the Veblen configuration on Fig.~\ref{fig:Desargues-Veblen-3d}, which results in the incorporation of the normalization map into the picture. Also different normalizations of equations for three points on the line imply the presence of the Newton map relations \eqref{eq:Newton-D}.
\begin{Rem}
In the Euclidean geometry approach the relation of the Schwarzian form of the discrete Kadomtsev--Petviashvili equation to the Menelaus theorem, closely related to the Veblen configuration, was given in~\cite{KoSchief-Men}. 
\end{Rem}
\begin{Prop}
The non-commutative Hirota system \eqref{eq:alg-cond}-\eqref{eq:shift-cond} is equivalent to the system composed of the normalization map \eqref{eq:S-bar-yx1}-\eqref{eq:S-bar-yx2} and of the Veblen map \eqref{eq:S-wu1}-\eqref{eq:S-wu2}.
\end{Prop}
\begin{proof}
The first part \eqref{eq:alg-cond} of the non-commutative Hirota is composed of three equations for different permutations of indices $i,j,k$. As it was shown in (Corollary 3.2 in \cite{Dol-Des}) any two of them imply the third one. We first show that equations  \eqref{eq:alg-cond} can be identified with equations \eqref{eq:S-bar-yx2} which form the second pair of the normalization map $\bar{S}$.

In the linear problem \eqref{eq:lin-A} identify $\bphi$ with $\bphi_A$ of equation \eqref{eq:lin-S-bar-12}, and $\bphi_{(i)}$ with $\bphi_B$, then $\bphi_{(k)}$ with $\bphi_C$, and $\bphi_{(j)}$ with $\bphi_D$. This gives the following identification of the coefficients:
\begin{equation*}
y_1 = - A_{ki}, \quad x_1 = A_{ik}, \quad y_2 = - A_{jk}, \quad x_2 = A_{kj}, \quad \bar{y}_2 = - A_{ji}, \quad \bar{x}_2 = A_{ij}.
\end{equation*}
Then equation \eqref{eq:alg-cond} of the non-commutative Hirota system is identified  with the second of equations \eqref{eq:S-Hir-alg}, while the first of equations \eqref{eq:S-Hir-alg} is \eqref{eq:alg-cond} with $i$ and $j$ reversed. The third equation of \eqref{eq:alg-cond} follows from the first two, as explained in \cite{Dol-Des}.
Equations \eqref{eq:S-bar-yx1} describe coefficients $\bar{y}_1$ and $\bar{x}_1$ of the linear equation involving $\bphi_{(i)}$, $\bphi_{(j)}$, and $\bphi_{(k)}$, which will be needed in the sequel.

In order to find the precise relation of the Veblen map to (three) equations of the second part \eqref{eq:shift-cond} of the Hirota system we first label the lines and points on Fig.~\ref{fig:Desargues-Veblen-3d} in a way described in Section~\ref{sec:V-pent}. By identifying the lines $L$ and $L_{(i)}$ with lines $1$ and $2$, correspondingly, and lines $L_{(j)}$ and $L_{(k)}$ with $1^\prime$ and $2^\prime$ we identify: $\phi_{(i)}$ with $\phi_{AC}$, $\phi_{(j)}$ with $\phi_{BC}$, $\phi_{(k)}$ with $\phi_{AB}$, $\phi_{(ij)}$ with $\phi_{CD}$, $\phi_{(ik)}$ with $\phi_{AD}$, and $\phi_{(jk)}$ with $\phi_{BF}$. This gives the following pairings:
\begin{align*}
w_2 & = - A_{jk(i)}^{-1},&  u_2 & = A_{kj(i)}A_{jk(i)}^{-1}, \\
w_1^\prime & = - A_{ki(j)}, & u_1^\prime & = A_{ik(j)}, \\
w_2^\prime & = - A_{ji(k)}^{-1}, & u_2^\prime &= A_{ij(k)} A_{ji(k)}^{-1},
\end{align*}
where the Newton map relations \eqref{eq:Newton-D} have been used as well.
The coefficients $w_1$ and $u_1$, which enter the linear equation involving $\bphi_{(i)}$, $\bphi_{(j)}$, and $\bphi_{(k)}$, can be obtained from the relation of the normalization map and the first part \eqref{eq:alg-cond} of the Hirota system, as described above. Due to equations \eqref{eq:S-bar-yx1} we have
\begin{equation*}
w_1 = \bar{y}_1 = (A_{ki} - A_{kj}) A_{jk}^{-1}, \qquad u_1 = \bar{x}_{1} = - A_{ik}A_{jk}^{-1}.
\end{equation*}

Then the second of equations \eqref{eq:S-wu1} is equation \eqref{eq:shift-cond} with the same indices $i,j,k$, while the first one defines the gauge coefficient
\begin{equation} \label{eq:G-H}
G = A_{ki(j)} A_{jk} A_{ik}^{-1}.
\end{equation}
The first equation of \eqref{eq:S-wu2} turns out to  be equation \eqref{eq:shift-cond} with indices $i$ and $k$ exchanged. Finally, in order to check that the second equation of \eqref{eq:S-wu2} is  equation \eqref{eq:shift-cond} with indices $j$ and $k$ exchanged (the last possibility), it is convenient to use the following consequence of equations \eqref{eq:alg-cond} proved in \cite{Dol-Des}
\begin{equation*}
A_{ij}^{-1} A_{ik} A_{jk}^{-1} + A_{kj}^{-1} A_{ki} A_{ji}^{-1} = 0.
\end{equation*}
\end{proof}

\section{The quantum Veblen map}
\label{sec:quant}
In this Section we discuss an example of the transition from solutions of the functional pentagon equation to its quantum version. This will be motivated by the quantization of the commutative (the division ring $\DD$ is replaced by the field $\CC$ of complex numbers) Veblen map. We show that, under appropriate choice of function $G$, the Veblen map $S^G$ preserves a natural Poisson structure, which is a limit of the Weyl commutation relations. After showing that the Veblen map preserves the commutation relations we construct a corresponding solution of the quantum pentagon equation. Analogous results for Zamolodchikov's tetrahedron equation~\cite{Zamolodchikov} in relation to the four dimensional consistency of the quadrilateral lattices~\cite{MQL} and discrete Darboux equations~\cite{BoKo} were discussed in~\cite{BaMaSe,Sergeev-q3w}.

Notice that because our evolution rule given by the Veblen map is valid for every division ring, from that point of view integrable quantization can be understood as an integrable reduction, where we specify a particular division ring as a ring of fractions of an algebra, and we demand the evolution preserves the structure of the algebra. This approach resolves also the problem of ordering of non-commutative factors, often present in the quantization procedures. Such a point of view is analogous to the idea of integrable geometric reductions of the quadrilateral lattices, as initiated in~\cite{CDS}. 

\subsection{The Veblen map as a Poisson map}
In this Section we investigate the Veblen map in the  commutative case, when the division ring $\DD$ is the field  $\CC$ of complex numbers. We will require however that the map is Poisson with respect to the brackets
\begin{equation} \label{eq:Poisson-b}
\{ w_i, u_j \} = \delta_{ij} w_i u_j, \qquad \{ w_i, w_j \} = \{ u_i, u_j \} = 0,
\qquad i,j=1,2.  
\end{equation}
This can be achieved by a suitable choice of the gauge parameter/function $G$ in formulas \eqref{eq:S-wu1}-\eqref{eq:S-wu2} of the map $S$.
\begin{Rem}
In other words we are looking for an automorphism of the Poisson algebra $\CC(w_1,u_1, w_2, u_2)$ of rational functions of four variables with the bracket defined by \eqref{eq:Poisson-b}.
\end{Rem}
\begin{Prop} \label{prop:S-G-P} 
If the gauge function $G:\CC^2\times\CC^2\dashrightarrow \CC^\times$ in the birational map $S^G:\CC^2\times\CC^2\dashrightarrow \CC^2\times\CC^2$ is of the form
\begin{equation} \label{eq:G-P}
G(w_1,u_1,w_2,u_2) = (\alpha u_2 + \beta w_1)w_2^{-1} u_1^{-1}, \qquad \CC^2 \ni (\alpha,\beta) \neq (0,0),
\end{equation}
then the map $S^G$ preserves the Poisson bracket \eqref{eq:Poisson-b}.
\end{Prop}
\begin{proof}
The result can be checked by direct calculation, which is however simpler when we 
make the ansatz $F(w_1,u_1,w_2,u_2) =G(w_1,u_1,w_2,u_2)u_1  w_2 $. Then the Poisson map condition gives relations
\begin{equation}
\frac{\partial F}{\partial u_1} = \frac{\partial F}{\partial w_2} = 0,
\end{equation}
and the homogeneity of degree one condition
\begin{equation}
w_1 \frac{\partial F}{\partial w_1} + u_2 \frac{\partial F}{\partial u_2} = F.
\end{equation}
\end{proof}
\begin{Cor} \label{cor:inv-P}
The inverse map \eqref{eq:S-inv-wu1}-\eqref{eq:S-inv-wu2} in the Poisson case with $G$ given by \eqref{eq:G-P} takes the form
\begin{align} \label{eq:S-P-inv-wu1}
w_1 & = w_2^\prime w_1^\prime , &  
u_1 & = \alpha u_2^\prime + \beta w_2^\prime u_1^\prime, \\
\label{eq:S-P-inv-wu2}
w_2 & = \alpha u_2^\prime (u_1^\prime)^{-1} + \beta w_2^\prime , & 
u_2 &= u_2^\prime w_1^\prime  (u_1^\prime)^{-1}.
\end{align}
The corresponding function $G^\prime$ reads
\begin{equation}
G^\prime (w_1^\prime, u_1^\prime, w_2^\prime, u_2^\prime) =
\frac{w_1^\prime}{\alpha u_2^\prime + \beta u_1^\prime w_2^\prime},
\end{equation}
and satisfies the functional relation 
\eqref{eq:G-prime}.
\end{Cor}
\begin{Rem}
From the proof it follows that the function $F$ can be of much more general form
\begin{equation*}
F(w_1,u_2) = f(w_1 u_2^{-1}) u_2 ,
\end{equation*}
with arbitrary function $f$.
\end{Rem}

In order to formulate the Poisson Veblen map analogue of Theorem~\ref{th:pentagon-S} we should translate the condition \eqref{eq:UV-XYZ}-\eqref{eq:UV-XYZ-funct} to conditions on the parameters $\alpha$ and $\beta$. The following result can be checked by direct calculation.
\begin{Lem}
The condition \eqref{eq:UV-XYZ}-\eqref{eq:UV-XYZ-funct} in the case of the gauge functions of the form given by equation \eqref{eq:G-P} is equivalent to the constraints
\begin{equation} \label{eq:relation-a-b-P}
\alpha_Y = \alpha_U \alpha_V, \quad  \alpha_X = \alpha_V \beta_Z, \quad
\alpha_Z = \alpha_U \beta_X,
\quad \beta_U = \beta_Y\beta_Z, \quad \beta_V = \beta_X\beta_Y.
\end{equation}
\end{Lem}
\begin{Rem}
As the set of independent parameters one can take $\alpha_U, \alpha_V$ and $\beta_X, \beta_Y, \beta_Z$.
\end{Rem}
\begin{Prop}
\label{prop:pentagon-S-P}
The  Poisson Veblen map $S(\alpha,\beta):\CC^2\times\CC^2 \dashrightarrow \CC^2\times\CC^2$ given by 
\begin{align} \label{eq:S-P-wu1}
w_1^\prime & = (\alpha u_2 + \beta w_1)w_2^{-1}   , & 
u_1^\prime & = w_2^{-1} u_1, \\
\label{eq:S-P-wu2}
w_2^\prime & = w_1 w_2(\alpha u_2 + \beta w_1)^{-1}, & 
u_2^\prime &= u_2 u_1(\alpha u_2 + \beta w_1)^{-1},
\end{align}
satisfies the functional pentagon relation
\begin{equation} \label{eq:pentagon-funct-dynam-P}
S_{12}(\alpha_U \beta_X, \beta_Z) \circ S_{13}(\alpha_U \alpha_V, \beta_Y)  \circ S_{23}(\alpha_V \beta_Z, \beta_X)  = S_{23}(\alpha_V, \beta_X \beta_Y)  \circ S_{12}(\alpha_U, \beta_Y \beta_Z) .
\end{equation}
\end{Prop}
\begin{Ex}
The solution $\alpha=0$, $\beta=1$ of equations \eqref{eq:relation-a-b-P} was discussed in the general non-commutative situation in Examples~\ref{ex:gauge} and~\ref{ex:gauge-2}.
\end{Ex}
\begin{Cor} \label{cor:locality}
Formulas \eqref{eq:S-P-wu1}-\eqref{eq:S-P-wu2} and \eqref{eq:S-P-inv-wu1}-\eqref{eq:S-P-inv-wu1} which define the map $S(\alpha,\beta)$ and its inverse have been obtained from general formulas \eqref{eq:S-wu1}-\eqref{eq:S-wu2}, \eqref{eq:S-inv-wu1}-\eqref{eq:S-inv-wu2} for $S^G$ and its inverse assuming only the form \eqref{eq:G-P} of $G$ (with $\alpha$ and $\beta$ in the center of $\DD$) and the ultra-locality, i.e. variables with different indices commute. 
\end{Cor}
\begin{Cor} \label{cor:norm-P}
Under the ultra-locality assumption the normalization map $\bar{S}$ given by equations \eqref{eq:S-bar-yx1}-\eqref{eq:S-bar-yx2}, with $w,u$ replaced by $y,x$, coincides with the local Veblen map \eqref{eq:S-P-wu1}-\eqref{eq:S-P-wu2} with parameters $\alpha=\beta=1$. Because such values of parameters respect conditions \eqref{eq:relation-a-b-P} then also the normalization map preserves in the commutative case the Poisson structure \eqref{eq:Poisson-b} and satisfies the functional pentagon equation.
\end{Cor}

\subsection{The Weyl reduction of the Veblen map}
Motivated by Proposition~\ref{prop:pentagon-S-P} we formulate its quantum analogue 
by replacing the Poisson brackets \eqref{eq:Poisson-b} by the corresponding Weyl commutation relations
\begin{gather} \label{eq:Weyl-cr}
 u_i w_j = q^{2\delta_{ij}} w_j u_i, \qquad
 w_i w_j = w_j w_i, \qquad  u_i u_j = u_j u_i , \qquad i,j=1,2 ,
\end{gather}
where we assume that $q$ is a non-zero element of the center of $\DD$. Because formulas \eqref{eq:S-P-wu1}-\eqref{eq:S-P-wu2} and \eqref{eq:S-P-inv-wu1}-\eqref{eq:S-P-inv-wu1} were derived under the locality assumption only we may ask if they preserve the commutation relations~\eqref{eq:Weyl-cr}. It turns out that such a simple choice works, which can be checked by direct calculation.  
\begin{Prop} \label{prop:S-Weyl}
The local Veblen map  $S(\alpha,\beta):\DD^2\times\DD^2 \dashrightarrow \DD^2\times\DD^2$, given by equations \eqref{eq:S-P-wu1}-\eqref{eq:S-P-wu2} with parameters $\alpha,\beta$, $(\alpha,\beta)\neq(0,0)$, from the center of $\DD$, preserves the Weyl commutation relations.
\end{Prop}
We have shown in particular, that the map \eqref{eq:S-P-wu1}-\eqref{eq:S-P-wu2} with entries satisfying relations \eqref{eq:Weyl-cr} preserves ultra-locality of the variables. It turns out that assuming only the ultra-locality of the input it is not enough to prove the ultra-locality of the output. We have however the following remarkable fact.
\begin{Prop}
Let $\Bbbk\subset \mathcal{Z}(\DD)$ be a subfield of the center of the division ring $\DD$ considered then as a $\Bbbk$-algebra. Denote by $\mathcal{A}_i$, $i=1,2$, $\Bbbk$-subalgebras of $\DD$ generated by elements $w_i, u_i\in\DD$, and by $\DD_i\subset \DD$ denote division hulls of $\mathcal{A}_i$. Assume:
\begin{enumerate}
\item ultra-locality of the algebras, i.e. $\mathcal{A}_i\subset \mathcal{C}_{\DD}(\mathcal{A}_j)$, $j\neq i$ is a subset of the centralizer of  $\mathcal{A}_j$ in $\DD$;
\item the general position condition, i.e. $\DD_1 \cap \DD_2 = \Bbbk$.
\end{enumerate}
If the algebras $\mathcal{A}_i^\prime$, $i=1,2$, generated by $w_i^\prime, u_i^\prime$ defined by equations \eqref{eq:S-P-wu1}-\eqref{eq:S-P-wu2} with $(\alpha,\beta)\in \Bbbk^2\setminus\{(0,0)\}$ satisfy also conditions (1) and (2) then there exists $q\in \Bbbk^\times$ such that the remaining part of the commutation relations \eqref{eq:Weyl-cr} is satisfied as well.
\end{Prop}
\begin{proof}
The condition $u_1^\prime u_2^\prime = u_2^\prime u_1^\prime$ and ultra-locality of the input variables in equations \eqref{eq:S-P-wu1}-\eqref{eq:S-P-wu2} leads to relation
\begin{equation}
w_1^{-1} u_1 w_1 u_1^{-1} = w_2 u_2^{-1}w_2^{-1} u_2,
\end{equation}
which due to the general position assumption concludes the proof (other ultra-locality relations for the output variables are identically satisfied or give equations equivalent to the above one).  
\end{proof}
\begin{Rem}
The analogous problem for the tetrahedron equation has been solved in \cite{Sergeev-super}.
\end{Rem}
Another question which should be answered before considering the pentagon condition in the case under consideration is if in the quantization of the Poisson structure of the map the relations \eqref{eq:relation-a-b-P} remain unchanged. Also here the answer is positive.
\begin{Lem}
The condition \eqref{eq:UV-XYZ} when applied to the gauge parameters of the form \eqref{eq:G-P} gives, in the case of Weyl commutation relations \eqref{eq:Weyl-cr}, the same constraints \eqref{eq:relation-a-b-P} on the parameters $\alpha$, $\beta$  like in the Poisson case.
\end{Lem}
\begin{proof}
 By direct verification, notice however that now the order of factors in equations \eqref{eq:G-P} and \eqref{eq:S-P-wu1}-\eqref{eq:S-P-wu2} matters. 
\end{proof}
\begin{Prop}
\label{prop:pentagon-S-q}
The  Veblen map $S(\alpha,\beta):\DD^2\times\DD^2 \dashrightarrow \DD^2\times\DD^2$ given by \eqref{eq:S-P-wu1}-\eqref{eq:S-P-wu2} when applied to the pairs $(w_i,u_i)$ constrained by the Weyl commutation relations \eqref{eq:Weyl-cr} with $i,j=1,2,3$,
satisfies the functional pentagon relation \eqref{eq:pentagon-funct-dynam-P}.
\end{Prop}

\begin{Rem}
Corollary~\ref{cor:norm-P} applies also when the Poisson structure \eqref{eq:Poisson-b} is replaced by the Weyl commutation relations~\eqref{eq:Weyl-cr}.
\end{Rem}

\subsection{Ring theoretical structures behind the Weyl commutation relations}
In this Section we briefly recapitulate known properties of the $q$-plane, which are relevant from the point of view of the paper. Moreover we sketch the idea how to get the solution of the quantum pentagon equation from the quantum Veblen map.

Let $\Bbbk$ be a field. Given $q\in\Bbbk^\times$, by $\mathcal{A}_q=\Bbbk\{ w, u\} /\mathcal{I}_q$ denote the standard Manin's quantum plane (see for example~\cite{Kassel}) defined as a quotient of the free associative $\Bbbk$-algebra with generators $w,u$ by the two sided ideal $\mathcal{I}_q$ generated by $uw - q^2 wu$. Its $N$-th tensor power $\mathcal{A}_q^{\otimes N}$ is isomorphic to the quotient algebra $\mathcal{A}_q(N) = \Bbbk\{ w_1,u_1, \dots ,w_N, u_N \} / \mathcal{I}_q(N)$, where the ideal $\mathcal{I}_q(N)$ is generated  by relations \eqref{eq:Weyl-cr} for  $1\leq i,j \leq N$. The isomorphism follows from identification $1\otimes \dots \otimes \stackrel{i}{w} \dots \otimes 1$ with $w_i$, and $1\otimes \dots \otimes \stackrel{i}{u} \otimes \dots \otimes 1$ with $u_i$, where $1 = 1_{\mathcal{A}_q}$. It is well known (see for example \cite{BrownGoodearl}) that $\mathcal{A}_q(N)$ is a Noetherian domain, and thus it has a division ring of fractions (quantum rational functions or quantum Weyl division algebra) denoted here by $\DD_q(N)$.
\begin{Rem}
In some theoretical physics' papers $\DD_q(N)$ is referred to as the Weyl algebra.
\end{Rem}
\begin{Rem}
For $q^2\neq 1$ the $q$-plane is rationally isomorphic to the (first) $q$-Weyl algebra $\tilde{\mathcal{A}}_q = \Bbbk\{ \tilde{w}, \tilde{u}\} /\tilde{\mathcal{I}}_q$, where two sided ideal $\tilde{\mathcal{I}}_q$ generated by 
$\tilde{u}\tilde{w} - q^2 \tilde{w} \tilde{u} - 1$ (recall that two Noetherian domains are called rationally isomorphic if  their division rings of fractions are isomorphic). The isomorphism is given by relations (see \cite{McConnellPettit})
\begin{equation*}
w = \tilde{w}, \qquad u = 1 +(q^2-1) \tilde{w} \tilde{u}, \qquad \tilde{u} = w^{-1}\frac{u-1}{q^2-1}.
\end{equation*} 
\end{Rem}
 
An intermediate object between $\mathcal{A}_q(n)$ and $\DD_q(N)$ is the corresponding $q$-torus $\mathcal{T}_q(N) =\Bbbk\{ w_1^{\pm1}, u_1^{\pm1}, \dots w_N^{\pm 1}, u_N^{\pm 1}\} /\mathcal{I}_q(N)$. It is known (see Proposition 1.3 of \cite{McConnellPettit}) that for $q$ not being a root of unity, which we assume in what follows, $\mathcal{T}_q(N)$ is simple, and its centre is $\Bbbk$. Elements of $\mathcal{T}_q(N)$ are Laurent polynomials in $w_1,\dots , u_N$, i.e. finite sums of $w_1^{i_1} u_1^{j_1} \dots  w_N^{i_N} u_N^{j_N}$, $(i_1,j_1,\dots ,i_N, j_N)\in\ZZ^{2N}$ with coefficients in $\Bbbk$. Another close related object is \cite{McConnellPettit} the 
Hahn--Mal'cev--Neumann completion \cite{Lam} of $\mathcal{T}_q(N)$, which will 
be denoted by $\bar{\DD}_q(N)$. Its elements are Laurent series with support 
bounded from below (we consider the lexicographic order in $\ZZ^{2N}$). It turns out that $\bar{\DD}_q(N)$ is a division ring containing 
$\DD_q(N)$ as a subring~\cite{Artamonov}.

By Proposition~\ref{prop:S-Weyl} we have:
\begin{Prop}
The following elements of $\DD_q(2)$
\begin{align} \label{eq:S-P-wu1-ten}
w_1^\prime & = (\alpha 1\otimes u + \beta w\otimes 1 ) 1 \otimes w^{-1}   , & 
u_1^\prime & = u \otimes w^{-1}, \\
\label{eq:S-P-wu2-ten}
w_2^\prime & = w \otimes w(\alpha 1\otimes u + \beta w\otimes 1)^{-1}, & 
u_2^\prime &= u\otimes u(\alpha 1\otimes u + \beta w\otimes 1)^{-1},
\end{align}
with $\Bbbk\ni (\alpha, \beta) \neq (0,0)$, generate a subalgebra $\mathcal{A}^\prime_q(2)$ of $\DD_q(2)$ isomorphic to $\mathcal{A}_q(2)$.
\end{Prop}
\begin{Rem}
The above Veblen map is an automorphism of the division ring $\DD_q(2)$.
\end{Rem}
\begin{Rem}
The $q$-torus $\mathcal{T}^\prime_q(2)$ obtained from $\mathcal{A}^\prime_q(2)$ is a subalgebra of  $\DD_q(2)$ isomorphic to $\mathcal{T}_q(2)$.
\end{Rem}
In the finite dimensional case, by the Skolem--Noether theorem~\cite{FarbDennis}, any isomorphism of two simple subalgebras of a central simple algebra is given by an inner automorphism of the algebra. Motivated by that we will be looking for an element $\boldsymbol{S}(\alpha,\beta)$ of a suitable completion of $\DD_q(2)$, which generates the transformation \eqref{eq:S-P-wu1-ten}-\eqref{eq:S-P-wu2-ten} as an inner automorphism. Then for any rational function $F$ of $(w_1,u_1,w_2,u_2)$ we would have
\begin{equation} \label{eq:S-hat}
\boldsymbol{S}(\alpha,\beta) F(w_1,u_1,w_2,u_2) \boldsymbol{S}(\alpha,\beta)^{-1} =
F(w_1^\prime,u_1^\prime,w_2^\prime,u_2^\prime) ,
\end{equation}
which, together with the functional pentagon relation \eqref{eq:pentagon-funct-dynam-P}, implies that, after eventual rescaling, $\boldsymbol{S}(\alpha,\beta)$ satisfies the quantum dynamical pentagon equation
\begin{equation} \label{eq:pentagon-quant-dynam-P}
\boldsymbol{S}_{23}(\alpha_V \beta_Z, \beta_X)
\boldsymbol{S}_{13}(\alpha_U \alpha_V, \beta_Y)
\boldsymbol{S}_{12}(\alpha_U \beta_X, \beta_Z) = 
\boldsymbol{S}_{12}(\alpha_U, \beta_Y \beta_Z)
\boldsymbol{S}_{23}(\alpha_V, \beta_X \beta_Y).
\end{equation}
\begin{Rem}
Notice that even in the root of unity case we are not in the finite dimensional situation because the map in general is not the identity map on the center. See the relevant discussion in \cite{Sergeev-finite} for the quantum tetrahedron equation.
\end{Rem}

\subsection{Quantization of the Veblen and normalization maps}
Let us consider the case $\Bbbk = \CC$ and equip the algebra $\mathcal{A}_q(2)$ with a $*$-structure, i.e. an involutive antiautomorphism $a\mapsto a^*$. Assume that $u^* = \mu u$ and $w^* = \nu w$, for $\mu,\nu \in\CC$, then it is easy to show that $q$, $\mu$ and $\nu$ must have modulus $1$. We will be interested when the Weyl reduction of the Veblen map preserves that additional Hermitean condition. By direct calculation we obtain the following result.
\begin{Prop} \label{prop:hermit}
The Weyl Veblen map is a $*$-map if and only if $(\alpha, \beta) \in \RR^2 \setminus\{(0,0\}$, and $\mu = q^{-2}$, $\nu = 1$.
\end{Prop}

The subsequent construction uses a special function which appeared first in works of Faddeev~\cite{Faddeev-LMP} under the name of quantum dilogarithm, and which is closely related to the Barnes double Gamma function~\cite{Barnes,Shintani}. Essentially the same function was used in the theory of non-compact quantum groups in~\cite{Woronowicz-q-exp} as the quantum exponential function. In various equivalent forms thie function was exploited in other works related to quantum integrable models \cite{Ruijsenaars-JMP,Kashaev,FKV-CMP,Kashaev-LCC,KLT-S,FoGo}.

The (non-compact) quantum dilogarithm is a meromorphic function defined by the integral representation 
\begin{equation} \label{eq-q-dil-int}
\varphi_b(z) = \exp \left(  \frac{1}{4} \int_{\RR+i0} \frac{\mathrm{e}^{-2izw}}{\sinh(wb) \sinh (w/b)} \frac{dw}{w} \right),
\end{equation}
where $b\in(0,\infty)$ is a parameter. It satisfies the difference equation
\begin{equation} \label{eq:q-dil-b}
\varphi_b(z - ib/2) = \left( 1 + \mathrm{e}^{2\pi z b} \right) \varphi_b(z + ib/2) ,
\end{equation}
which is used to extended it, from the initial convergence strip of the integral representation \eqref{eq-q-dil-int}, to a meromorphic function on the whole complex plane $\CC$. Its poles (zeros) are at points 
\begin{equation*}
(-)(i\eta + i m b + i n b^{-1}), \qquad m,n\in\NN_0 , \qquad \eta = \frac{1}{2}\left( b + b^{-1}\right).
\end{equation*}
If $b$ is real (or a pure phase $|b|=1$, which is not our case) the function $\varphi_b(z)$ is unitary in the sense that
\begin{equation}
\overline{\varphi_b(z)} = 1/\varphi_b(\bar{z}).
\end{equation}
For detailed discussion of its other remarkable properties we refer to the cited works, see also \cite{Ruijsenaars-JNMP} for a review.  

Let $x$ and $p$ provide the  Schr\"{o}dinger representation in $\mathcal{H}=L^2(\RR)$ of the Heisenberg commutation relations
\begin{equation}
[ x, p ] = \frac{i}{2\pi}.
\end{equation}
We have then the corresponding representation of the Weyl pair which respects the Hermitean restriction described in Proposition~\ref{prop:hermit}
\begin{equation} \label{eq:UW}
U = q\mathrm{e}^{2\pi b x},\qquad W = \mathrm{e}^{2\pi b p}, \qquad U W = \mathrm{e}^{2 \pi i b^2} W U, \qquad q=\mathrm{e}^{\pi i b^2}.
\end{equation}
\begin{Rem}
The additional requirement $b^2 \not\in \QQ$, which corresponds to $q^2$ being root of unity is not essential in this representation.
\end{Rem}
In the class of functions which can be analytically extended around the strip $0\leq |\mathrm{Im} \, z | \leq b$ the action of $W$ reads 
\begin{equation} \label{eq:W-f}
W^{-1} f(x) W = f(x+ib).
\end{equation} 
There is remarkable five term relation \cite{FKV-CMP,Woronowicz-q-exp} connecting the non-compact quantum dilogarithm function and Heisenberg pairs
\begin{equation} \label{eq:five-term}
\varphi_b(p)\varphi_b(x) = \varphi_b(x) \varphi_b(p+x) \varphi_b(p),
\end{equation}
which will be used in the proof of the final result of the paper.

\begin{Prop}
For real $\alpha$, $\beta$ let $\boldsymbol{S}(\alpha,\beta)$  be the unitary operator on $\mathcal{H}\otimes \mathcal{H}$ given by
\begin{equation} \label{eq:b-S}
\boldsymbol{S}(\alpha,\beta) = \mathrm{e}^{i\pi \mu (\mu - \lambda)} \varphi_b(\lambda - \mu + 1\otimes x - p\otimes 1) \mathrm{e}^{2 \pi i \mu (1\otimes x - p\otimes 1)} \mathrm{e}^{2 \pi i(x - p )\otimes p},
\end{equation}
where $\alpha = \mathrm{e}^{2\pi b \lambda}$, and $\beta = \mathrm{e}^{2\pi b \mu}$. Then:\\
1) $\boldsymbol{S}(\alpha,\beta)$ is the intertwiner satisfying \eqref{eq:S-hat},
in particular we have 
\begin{align} \label{eq:bS-P-w1-ten}
\boldsymbol{S}(\alpha,\beta) ( W\otimes 1) \boldsymbol{S}(\alpha,\beta)^{-1} & = (\alpha 1\otimes U + \beta W\otimes 1 ) (1 \otimes W^{-1})   , \\
\label{eq:bS-P-u1-ten}
\boldsymbol{S}(\alpha,\beta) (U\otimes 1) \boldsymbol{S}(\alpha,\beta)^{-1}  & = U \otimes W^{-1}, \\
\label{eq:bS-P-w2-ten}
\boldsymbol{S}(\alpha,\beta) (1\otimes W) \boldsymbol{S}(\alpha,\beta)^{-1}  & = (W\otimes W) (\alpha 1\otimes U + \beta W\otimes 1)^{-1}, \\
\label{eq:bS-P-u2-ten}
\boldsymbol{S}(\alpha,\beta) (1\otimes U) \boldsymbol{S}(\alpha,\beta)^{-1}  &= (U\otimes U)(\alpha 1\otimes U + \beta W\otimes 1)^{-1}.
\end{align}
2) It satisfies the 
quantum dynamical pentagon equation
\eqref{eq:pentagon-quant-dynam-P}.
\end{Prop}
\begin{proof}
To show the first point we use property \eqref{eq:q-dil-b} of the function $\varphi_b$, the representation \eqref{eq:W-f} of the shift action of the exponentiated momentum operator, and the Baker--Campbell--Hausdorff (BCH) formula. Details are known and can be found, for example, in \cite{FrenkelKim}. 

In order to conclude the proof of the second point (the transition from the functional to quantum pentagon equation has been discussed earlier) we have to check the scalar factor in front of the operator part of $\boldsymbol{S}(\alpha,\beta)$. We will do it however in a way independent of the previous part by showing directly that the intertwiner satisfies the quantum dynamical pentagon equation. Denote by 
\begin{equation*}
\sigma = \mathrm{e}^{2 \pi i \mu (1\otimes x - p\otimes 1)} \mathrm{e}^{2 \pi i(x - p )\otimes p}
\end{equation*}
the third part in the structure of $\boldsymbol{S}$ apart the normalization and dilogarithmic parts. After shifting all $\sigma_{jk}$ in the pentagon equation to the right and cancelling the dilogarithmic part of the equation using the five term relation \eqref{eq:five-term} with
\begin{equation*}
x = \lambda_V + \mu_Z - \mu_X + x_3 - p_2, \qquad p = \lambda_U - \mu_Y - \mu_Z + x_2 - p_1,
\end{equation*}
we are left with the equation
\begin{equation} \label{eq:qP-red}
\sigma_{23}^X \sigma_{13}^Y \sigma_{12}^Z = \mathrm{e}^{i\pi (\mu_Y^2 + 2 \mu_X \mu_Y + 2 \mu_X \mu_Z +2 \mu_Y \mu_Z)} \sigma_{12}^U \sigma_{23}^V,
\end{equation}
which can be proven by BCH exchanges.
\end{proof}
\begin{Cor}
The Weyl reduction of the normalization map (see Remark after 
Proposition~\ref{prop:pentagon-S-q} and Corollary~\ref{cor:norm-P}) corresponds to the case $\lambda=\mu=0$ in formula \eqref{eq:b-S}. In an equivalent form such a solution of the quantum pentagon equation was used by Kashaev~\cite{Kashaev-LCC} in the quantum Teichm\"{u}ller theory.
\end{Cor}
\begin{Rem}
Essentially the same formula as \eqref{eq:b-S} giving a solution of the quantum pentagon equation was obtained in \cite{Sergeev-book} as a limit of a solution of Zamolodchikov's tetrahedron equation.
\end{Rem}

\section{Conclusion and perspectives}
The same elementary incidence geometric structures, the Veblen and Desargues configurations, which are behind the non-commutative discrete KP equation and its integrability, provide solutions to the pentagon equation. There is a natural question how large is the class of the geometric (in the above sense) solutions of the functional or quantum pentagon equation within all the solutions. We note that in the commutative case, as it was demonstrated recently \cite{ABS-octahedron}, the discrete KP equation (in various equivalent forms) is the only multidimensionally consistent equation of octahedron type.

This corresponds to the already mentioned observation \cite{BaMaSe-P} that all presently known solutions of the quantum tetrahedron equation can be obtained, by a canonical quantization, from classical solutions of the functional tetrahedron equation describing \cite{BaMaSe,Sergeev-q3w} four dimensional consistency of quadrilateral lattices \cite{MQL}. It would be interesting to study relation of the quantum pentagon and tetrahedron equations, as initiated in \cite{Kashaev-Sergeev}. As it was shown in \cite{Dol-Des}, on the division ring level Desargues maps are equivalent to quadrilateral lattices, provided we take also their Laplace transforms \cite{DCN,TQL} into consideration.

On the level of commutative discrete three dimensional multidimensionally consistent systems there are not known other then the geometric ones. These are the Hirota system, describing Desargues maps, and its reductions:
\begin{itemize}
\item discrete Darboux equations \cite{BoKo} describing multidimensional quadrilateral lattices \cite{MQL}, whose relation to Desargues maps was given in \cite{Dol-Des};
\item the Miwa (discrete B-KP equation) \cite{Miwa} describing the B-quadrilateral lattices~\cite{BQL};
\item discrete C-KP equation~\cite{Kashaev-LMP,Schief-JNMP,TsarevWolf} describing the C- (or symmetric) quadrilateral lattices \cite{DS-sym,CQL}.
\end{itemize}

As far as two dimensional integrable systems and integrable maps are concerned, there is abundance of examples of their derivation as reductions of the Hirota equation, directly or via the three dimensional systems listed above, see the very incomplete list \cite{KNW-KP,KNW-BKP,Noumi,GRT} which starts from the original paper \cite{Hirota} of Hirota. In fact, the possibility of writing of a given system in Hirota-like bilinear form is a paradigm of the integrable system theory.

It is remarkable that the the Veblen flip, which turned out to be so fundamental in the domain of integrable discrete systems, supplemented by the very physical ultra-locality condition singles out the Weyl commutation relations, modulo the possible different choice of the gauge function $G$, which however does not affect the geometric structure.

\section*{Acknowledgments}
A. D. would like to thank Piotr Stachura and an anonymous Referee for their criticism, which helped to improve presentation. The work of A. D. was supported in part by the Ministry of Science and Higher Education grant No. N~N202~174739.

\bibliographystyle{amsplain}

\providecommand{\bysame}{\leavevmode\hbox to3em{\hrulefill}\thinspace}

\end{document}